\documentclass[11pt,reqno]{amsart}
\usepackage{amsthm,amsmath,amssymb,hyperref}
\usepackage{color}

\newtheorem{theorem}{Theorem}[section]

\newtheorem{lemma}[theorem]{Lemma}
\newtheorem{proposition}[theorem]{Proposition}
\theoremstyle{remark}
\newtheorem{remark}[theorem]{Remark}

\DeclareRobustCommand{\leanlink}[1]{%
    \href
    {#1}
    {\ensuremath{\forall}}
}
\author[Durcik]{Polona Durcik}
\address{Schmid College of Science and Technology, Chapman University, Orange, CA, USA}
\email{durcik@chapman.edu}

\author[Fraccaroli]{Marco Fraccaroli}
 \address{Department of Mathematics and Statistics, 
	University of Massachusetts Lowell,
	Lowell, MA, USA}
\email{marco\_fraccaroli@uml.edu}

\author[Roos]{Joris Roos}
 \address{Department of Mathematics and Statistics, 
	University of Massachusetts Lowell,
	Lowell, MA, USA}
\email{joris\_roos@uml.edu}

\numberwithin{equation}{section}

\title[Weak Hellinger]{A weak Hellinger inequality for noisy Boolean channels}
\date{\today}

\begin{document}

\begin{abstract}
A weak form of the Hellinger conjecture of Anantharam, Bogdanov, Chakrabarti, Jayram, and Nair for the binary symmetric channel is proved: dictator functions maximize Hellinger $\Phi$-entropy among all Boolean functions of the input and all one-bit statistics of the output of a noisy channel.
The technical heart of the matter is an explicit inequality in three real parameters, which is proved using explicit polynomial approximations and computer-assisted positivity checks.  The results are also formally verified in Lean 4.
\end{abstract}

\maketitle

\section{Introduction}\label{sec:intro}

Let $n\ge 1$ and let $X=(X_1,\dots,X_n)$ be uniform on $\{-1,1\}^n$.
Fix $\rho\in[-1,1]$ and let $Y$ be obtained from $X$ by passing each coordinate through a binary symmetric channel with correlation parameter $\rho$, so that $Y_i=X_i Z_i$ where $Z_1,\dots,Z_n$ are independent with $\Pr(Z_i=1)=\tfrac{1+\rho}{2}$ and $\Pr(Z_i=-1)=\tfrac{1-\rho}{2}$. In this case we also say that $X$ and $Y$ are $\rho$-correlated and write $X\sim_\rho Y$.

The Courtade--Kumar conjecture asserts that for every Boolean function $f$ from $\{-1,1\}^n$ to $\{-1,1\}$ one has
\begin{equation}\label{eq:courtadekumar}
I(f(X);Y)\;\le\;I(X_1;Y_1),
\end{equation}
with equality attained by dictator functions $f(X)=X_i$ \cite{CK14}. Here $I(X; Y)$ denotes mutual information.   Thus the conjecture can be understood, roughly speaking, as saying that dictator functions maximize mutual information on noisy inputs.
The conjecture is known in the high-noise regime (equivalently, for $|\rho|$ sufficiently small): Samorodnitsky proved that there exists an absolute constant $\delta>0$ such that the conjecture holds whenever $|\rho|\le 2\delta$ \cite{Sam15}.
More recently, Javanmard and Woodruff sharpened the high-noise entropy expansion and extended the range of noise parameters for which the conjecture can be verified \cite{JW26}.
Beyond high noise the conjecture remained open until very recently\footnote{After completing this manuscript, we learned of the very recent preprints \cite{CGJLMNW26,KT26,MB26}, which announce proofs of the Courtade--Kumar conjecture.}; see also the differential-equation approach of Chen, Gohari, and Nair \cite{CGN25}.

A stronger inequality, proposed by Anantharam, Bogdanov, Chakrabarti, Jayram, and Nair \cite{ABCNJ17} and referred to as the Hellinger conjecture, takes the form
\begin{equation}\label{eq:hellinger}
\sqrt{1-(\mathbf{E}f(X))^2}\;-\;\mathbf{E}\sqrt{1-(\mathbf{E}[f(X)\mid Y])^2}
\;\le\;1-\sqrt{1-\rho^2},
\end{equation}
where the quantity on the left-hand side can be written as $\Phi$-entropy in the sense of \cite{ABCNJ17} with $\Phi$ being the squared Hellinger distance. Roughly speaking, the left-hand side measures how much ``Hellinger information'' is revealed about $f(X)$ by knowing the full noisy output $Y$, and the conjecture states that this quantity is maximized by taking $f$ to be a dictator function.
As noted in \cite{ABCNJ17}, \eqref{eq:hellinger} would imply the Courtade--Kumar conjecture \eqref{eq:courtadekumar} by convexity.
The Hellinger conjecture, its weak form, and the three-variable inequality below were originally formulated during the Simons Institute program on Information Theory\footnote{C.~Nair, personal communication} in Spring 2015.
Their motivation was an earlier conjecture of Anantharam, Gohari, Kamath, and Nair \cite[Conjecture~4]{AGKN13Allerton}, relating the mutual information of a binary pair to the hypercontractivity parameter at infinity.

The Hellinger conjecture also has an interesting prediction in the low-noise limit.
For balanced $f$ and $\rho\uparrow 1$, it yields a sharp inequality at exponent $\tfrac12$ for sensitivity:
\begin{equation}\label{eq:sensitivity}
\mathbf{E}\sqrt{s_f(X)}\;\ge\;1,
\end{equation}
where $s_f(x)$ is the sensitivity of $f$ at $x$, i.e. the number of Hamming neighbors of $x$ on which $f$ changes value.
The inequality \eqref{eq:sensitivity} has recently been proved as a consequence of
a certain sharp isoperimetric inequality on the Hamming cube, see \cite{DIRX26} (also see \cite{DIR24,BIM23,BG99,Tal95}).

In the present note we prove a weak form of the Hellinger conjecture, stated in \cite{ABCNJ17} (Conjecture 3), in which conditioning on the full channel output $Y$ is replaced by conditioning on an arbitrary one-bit statistic $g(Y)$.

\begin{theorem}[Weak Hellinger \leanlink{https://github.com/roos-j/lean-weakhellinger/blob/1ef447e6fc6027b97f7b1c8c2ae10ebafeaab5f8/WeakHellinger/Theorems.lean\#L26-L33}]\label{thm:main}
Let $n\ge 1$ and $\rho\in[-1,1]$. Let $X$ be a uniformly distributed random variable on $\{-1,1\}^n$ and $X\sim_\rho Y$.
Then for all Boolean functions $f,g:\{-1,1\}^n\to\{-1,1\}$,
\begin{equation}\label{eq:weak-hellinger}
\sqrt{1-(\mathbf{E}f(X))^2}\;-\;\mathbf{E}\sqrt{1-(\mathbf{E}[f(X)\mid g(Y)])^2}
\le 1-\sqrt{1-\rho^2}.
\end{equation}
\end{theorem}

The proof follows a reduction, due to \cite{ABCNJ17}, to an explicit inequality in three real parameters. 
The key of this reduction is to pass through a stronger inequality involving the hypercontractivity parameter at zero associated with the two binary random variables, see \eqref{eqn:stronger} (Conjecture 5 in \cite{ABCNJ17}).
The three-point inequality is proved in \S \ref{sec:threepoint} using careful approximations by Bernstein polynomials and explicit computations verifying a large number of positivity claims.
For completeness we also give the details of the reduction in \S \ref{sec:reduction}.

As a corollary of this result we also recover the known corresponding statement for mutual information (see \cite{AGKN13Allerton}, \cite{PPM18}):
\begin{equation}\label{eqn:weakmutualinf}
I(f(X);g(Y))\le 1-H_2((1-|\rho|)/2),
\end{equation}
where $H_2(t)=-t\log_2 t-(1-t)\log_2(1-t)$ denotes binary entropy and $X,Y,f,g,\rho$ are as in Theorem \ref{thm:main}.
In fact, a stronger statement involving the hypercontractivity parameter at zero also holds, see Remark \ref{rem:mutualinf} below.
We thank Chandra Nair for pointing this out to us.

\subsection{Formalization in Lean}
Numerical evidence for the correctness of the three-point inequality is easy to produce.
The main contribution of this paper is a definitive formal verification, both via traditional mathematical arguments and using the proof assistant Lean 4 \cite{LeanPaper} and its mathematical library \texttt{mathlib} \cite{Mathlib}.
The formalization is available on GitHub at 
\begin{center}
\href{https://github.com/roos-j/lean-weakhellinger}{https://github.com/roos-j/lean-weakhellinger}
\end{center}
It comprises $\sim$500k lines of Lean code and compiling the theorem takes $\sim$40 GB of available RAM and $\sim$2.5 hours on one of the authors' machines.
This is because all explicit computations in the paper (which take just a few seconds to run on usual hardware) must be encoded into Lean kernel terms for typechecking. This comes at an enormous computational cost, despite various optimizations. Most importantly, all computations for the verification use only rational numbers, completely avoiding the need for floating point numbers and interval arithmetic. The Lean kernel has special handling for natural numbers which makes verifications like this efficient enough to be performed using a single consumer-grade machine, even if barely so.

The main statements are annotated with links to the respective Lean statements (\ensuremath{\forall}). The Lean statements are human-written, while the proofs are machine-generated.

\subsection*{Statement on use of AI}
This manuscript was written by the authors. OpenAI’s GPT models played a substantial role in the development of the proof of the three-point inequality through extended discussions with the authors. The formalization of the proofs in Lean was completed using Codex/gpt-6-astra-medium under human supervision.

\subsection*{Acknowledgements}
This work was supported in part by NSF grants DMS-2554859 (P.D.), DMS-2154835 (J.R.), DMS-2555784 (J.R.), and grants from the Simons Foundation.
M.F. was supported in part by an AMS-Simons Travel Grant. The authors also acknowledge the Hausdorff Institute of Mathematics in Bonn, where they first learned about this problem at the Dual Trimester program ``Boolean Analysis in Computer Science'' during Fall 2024 from a talk given by Chandra Nair.
The authors also thank Chandra Nair for helpful discussion of this and related problems, and helpful comments on a previous version of this paper, in particular suggesting \eqref{eqn:weakmutualinf} and Remark \ref{rem:mutualinf}. This paper is also an indirect outgrowth of a SQuaRE workshop on related problems. P.D. and J.R. thank the American Institute of Mathematics for supporting the workshop and their fellow SQuaRE members Irina Holmes, Paata Ivanisvili, and Alexander Volberg.

\section{A three-point inequality}\label{sec:threepoint}

For $x\in [-1,1]$ let $\bar{x} = 1-x$ and $h(x) = \sqrt{1-x^2}$. 
For $u,v\in (0,1)$, define the binary Kullback--Leibler divergence
\[
D(u\mid v)=u\log\frac uv+\bar u\log\frac{\bar u}{\bar v}
\]
with the usual extension to boundary values in the extended real numbers, in particular $D(p|0)=D(p|1)=\infty$, $D(0|q)=\log \frac{1}{\bar q}$, $D(1|q)=\log\frac1q$ and $D(0|0)=D(1|1)=0$.

Define $U=\{(s,c,d)\in(0,1)^3\,:\,c\not=d\}\subset [0,1]^3$ and
\[ R(s,c,d) = \sqrt{\frac{D\!\big(s\bar c+\bar s\,d\mid s\bar d+\bar s\,c\big)}
{s\,D(c\mid d)+\bar s\,D(d\mid c)}} \]
for $(s,c,d)\in U$, and 
with the understanding that $R(s,c,d)$ is defined as 
\[ R(s,c,d)= \limsup_{U\ni(s',c',d')\to (s,c,d)} R(s',c',d')\]
for $(s,c,d)\in [0,1]^3\setminus U$.
Finally, define for $(s,c,d)\in [0,1]^3$,
\[ F(s,c,d) = 1 - h\big(s(\bar d-d)+\bar s(c-\bar c)\big)  +\;
s\,h(\bar d-d) + \bar s\, h(c-\bar c) - R(s,c,d).\]

We have the following theorem. 
\begin{proposition}[Three-point inequality \leanlink{https://github.com/roos-j/lean-weakhellinger/blob/1ef447e6fc6027b97f7b1c8c2ae10ebafeaab5f8/WeakHellinger/Theorems.lean\#L15-L18}]\label{prop:three-point}
For all $s,c,d\in[0,1]$,
\begin{equation}\label{eq:three-point}
F(s,c,d)\ge 0.
\end{equation}
Equality holds if and only if $s\in \{0,1\}$, or $c+d=1$, or $s=\frac12$ and $c = d$.
\end{proposition}

We introduce the auxiliary function
\[
 K(u \mid   v)=\left\{\begin{array}{ll}
 D(u\mid  v)(u-v)^{-2}   & u\ne v, \\
(2v(1-v))^{-1}  & u=v.
 \end{array} \right .
\]
The value at $u=v$ is set such that $K$ is continuous across the diagonal $u=v$. 
For $c,d\in(0,1)$ and $s\in[0,1]$, we have
\[
 R(s,c,d)=\sqrt{\frac{K(s\bar c+\bar s d\mid  s\bar d+\bar s c)}{sK(c\mid  d)+\bar sK(d\mid  c)}}.
\]
Furthermore, if at least one of $c,d$ is $0$ or $1$, we have
\[
 R(s,c,d)=\left\{\begin{array}{ll}
 1   & \text{if
$s\in\{0,1\}$ or $c+d=1$,} \\
0  & \text{otherwise.}
 \end{array} \right .
\]

We begin the proof of Proposition \ref{prop:three-point} with the following integral representation of the quantity $K$.

\begin{lemma}\label{lem:K-integral}
For all $u,v\in(0,1)$,
\[
 K(u\mid  v)=\int_0^1\frac{1-\lambda}{((1-\lambda)v+\lambda u)\,(1-(1-\lambda)v-\lambda u)}\,\mathrm{d}\lambda .
\]
In particular,  $K$ is
 positive and continuous on $(0,1)^2$.
\end{lemma}
\begin{proof}
    
 For $x\in (0,1)$, the function $\varphi(x)=x\log x+(1-x)\log(1-x)$ is smooth with 
 \[
 \varphi'(x)=\log\frac{x}{1-x},\qquad
 \varphi''(x)=\frac1x+\frac1{1-x}=\frac1{x(1-x)},
\]
and 
\[D(u\mid v)=\varphi(u)-\varphi(v)-\varphi'(v)(u-v).\]
By the Taylor's formula with integral remainder, \[
 D(u\mid  v)=(u-v)^2\int_0^1 (1-\lambda)\,
 \varphi''\big((1-\lambda)v+\lambda u\big)\,\mathrm{d}\lambda\]
 \[
 =(u-v)^2\int_0^1\frac{1-\lambda}
 {((1-\lambda)v+\lambda u)\,(1-(1-\lambda)v-\lambda u)}\,\mathrm{d}\lambda .
\]
For $u\ne v$, division by $(u-v)^2$ proves the formula. For $u=v$, the
integral equals $(2v(1-v))^{-1}$, as required by the definition of $K$.
Positivity and continuity follow from the integral representation. 
\end{proof}

Note that the desired inequality \eqref{eq:three-point} is equivalent to
\[ h(2q-1)-s\,h(2d-1)-\bar s\,h(2c-1)+R(s,c,d)\leq1.\]
Suppose first that at least one of $c,d$ belongs to $\{0,1\}$. If
$s\in\{0,1\}$ or $c+d=1$, then $R=1$, and direct substitution gives equality
in \eqref{eq:three-point}. Otherwise $R=0$ and $s\in(0,1)$, so the value of $F(s,c,d)$ is at least $1 - h(2q-1) \geq 0$. Equality can hold only if
$q=\tfrac12$ and $h(2c-1)=h(2d-1)=0$. Since $c+d\ne1$, this means
$c=d\in\{0,1\}$. In these two cases $q=s$ and $q=1-s$, respectively, so
equality holds exactly when $s=\tfrac12$. Thus the boundary cases, including
their equality cases, follow. So from now on assume  $c,d\in(0,1)$.

  Introduce new coordinates
\[
 a=c+d-1,\qquad b=d-c,\qquad t=1-2s.
\]
Equivalently, $c=\tfrac{1}{2}(1+a-b)$, $d=\tfrac{1}{2}(1+a+b)$, $s=\tfrac{1}{2}(1-t)$, $\bar s=\tfrac{1}{2}(1+t)$, i.e. 
\[
 2c-1=a-b,\qquad 2d-1=a+b,\qquad 2p-1=u+b,\qquad 2q-1=u-b, 
\]
and that
$c,d\in (0,1)$ if and only if  $|a|+|b|<1$.  
Moreover, $a=0$ corresponds to $c+d=1$, $|t|=1$ to $s\in\{0,1\}$, $b=0$ to $c=d$, and $t=0$ to $s=\tfrac12$.

The maps  $(s,c,d)\mapsto(\bar s,d,c), (s,c,d)\mapsto(s,\bar c,\bar d)$
send $(a,b,t)$ to $(a,-b,-t)$ and $(-a,-b,t)$, respectively, and they preserve the boundary conditions   $|t|=1$ and $a=0$. They   leave  the value of $F(s,c,d)$ invariant, and their compositions allow us to assume 
\[
 a\geq0,\qquad b\geq0,\qquad a+b<1,\qquad -1\leq t\leq1.
\]
The three equality cases $a=0$, $|t|=1$, and $b=t=0$ are preserved by these transformations.

For $|x|+b<1$ set
\[
 k(x,b)=K\big(\tfrac{1}{2}(1+x+b)\,\big|\,\tfrac{1}{2}(1+x-b)\big)
 =\int_{-1}^1\frac{1-v}{1-(x+bv)^2}\,\mathrm{d} v .
\]The condition $|x|+b<1$ ensures $|x+bv|<1$ for  $x\in [-1,1]$, so $k(x,b)>0$.
Denote
\[\widetilde{R}(a,b,t) = R(\tfrac{1}{2}({1-t},{1+a-b},{1+a+b})),\]
i.e. the quantity   $R$ expressed in the new variables. Then
\[     \widetilde{R}^2(a,b,t)=\frac{k(at,b)}{\beta(a,b,t)}, \]
where
\begin{equation}
    \label{eq:r2b}
    \beta(a,b,t)=\frac{1+t}{2}\,k(a,b)+\frac{1-t}{2}\,k(-a,b).
\end{equation}
Denote also
\begin{equation}
    \label{eq:eC}
    \gamma(a,b,t)=\frac{1+t}{2}\,h(a-b)+\frac{1-t}{2}\,h(a+b), 
\end{equation}
and
  \[\delta(a,b,t)=1-h(at-b)+\gamma(a,b,t) .\]
Then  \eqref{eq:three-point} is  equivalent to
\begin{equation}
    \label{eq:CR}   \delta(a,b,t)\ge \widetilde{R}(a,b,t)
\end{equation}
for all $a,b\ge 0$, $a+b<1$, and  $-1\le t\le 1$.

Direct substitution gives $\widetilde{R}=\delta=1$
if $|t|=1$ or $a=0$.
 If $b=0$, then $k(a,0)=2/h(a)^2$, so $\beta=2/h(a)^2$, $\gamma=h(a)$, and
\begin{equation}\label{eq:diagonal}
 \widetilde{R}(a,0,t) =\frac{h(a)}{h(ta)},\qquad
 \delta(a,0,t)-\widetilde{R}(a,0,t)
 =\frac{\bigl(1-h(ta)\bigr)\bigl(h(ta)-h(a)\bigr)}{h(ta)}\ \ge0 .
\end{equation}
Here $0<h(a)\leq h(ta)\leq1$ and equality holds exactly when $ta=0$ or $|ta|=a$,
that is, $a=0$, $t=0$, or $|t|=1$. 

Let now $a,b>0$ and $|t|<1$. Since $k, \beta$, and $\delta$ are  positive,  \eqref{eq:CR} is equivalent to
\begin{equation}\label{eq:C2B}
 \delta^2(a,b,t)\beta(a,b,t)\ \ge\ k(at,b). 
\end{equation}
Note  that the right-hand side depends on $(a,t)$ only through the product  $ta$, which motivates the following definition.

For
  $a>0,\ b>0,\ a+b<1,\ |u|\le a$, define
   \[G(a,b,u) = \delta^2 (a,b,u/a) \beta (a,b,u/a) \]
   \[ = \Big( 1-h(u-b) + \frac{a+u}{2a}\,h(a-b)+\frac{a-u}{2a}\,h(a+b) \Big)^2  \Big (\frac{a+u}{2a}\,k(a,b)+\frac{a-u}{2a}\,k(-a,b) \Big).  \]
It then suffices to prove the following monotonicity statement. 
\begin{proposition}\label{prop:main}
For all $a,b>0$, $a+b <1$,  and all $u\in \mathbb{R}$ with $|u|<  a$, 
   \[\partial_a G(a,b,u) >0.   \]
\end{proposition}
The desired inequality \eqref{eq:C2B} is then obtained as an immediate consequence.
\begin{proposition}
    For $a>0$, $b>0$, $a+b<1$ and $|t|<1$, we have $\delta^2(a,b,t)\beta(a,b,t) >k(at,b)$, and hence the strict inequality in \eqref{eq:CR}. 
\end{proposition}
\begin{proof}
Fix $b,u$ with $|u|+b<1$. Proposition~\ref{prop:main} shows that
$a\mapsto G(a,b,u)$ is strictly increasing for $|u|<a<1-b$. If $u\ne0$,
$F$ is continuous at $a=|u|$, and the case $t=u/a=\pm1$ gives
$G(|u|,b,u)=k(u,b)$. If $u=0$, continuity of $r$ and $k$ gives
\[
 \lim_{a\to 0+}G(a,b,0)=k(0,b).
\]
Hence $G(a,b,u)>k(u,b)$ whenever $a>|u|$. The claim then follows by choosing $u=ta$. 
\end{proof}

For $0\leq x<1$, we introduce 
\[
 L(x) =\int_0^1\frac{\mathrm{d}\lambda}{1-x \lambda^2}.
\]
We have $L(0)=1$ and for $x >0$, by integrating, 
    \begin{equation}\label{eq:H1_expr}
    L(x)=\frac1{2\sqrt x}\log \frac{1+\sqrt{x}}{1-\sqrt{x}} .
    \end{equation}
In the following lemmas, we collect some auxiliary approximation properties that will be used in the proof of Proposition~\ref{prop:main}.
\begin{lemma} 
    For all $x\in [0,1)$, 
    \begin{equation}\label{eq:H2_bounds}
 1+\frac x3\leq L(x)\leq
 \frac{7+x}{15}+\frac8{15\sqrt{1-x}},
 \end{equation}
 and for all $0\le y\le x < 1$, 
 \begin{equation}\label{eq:H3_der}
 L(x) - L(y) \leq\frac12\log\frac{1-y}{1-x}-\frac{x-y}{6}.
 \end{equation}
\end{lemma}

\begin{proof}
The case $x=0$ is immediate, so we assume $x > 0$.

To prove the lower bound in \eqref{eq:H2_bounds}, we integrate the inequality
\[
(1-x\lambda^2)^{-1}\geq1+x\lambda^2.
\]
To prove the upper bound, we consider the function 
\[
 g(v)=v \Big(\frac{7+v^2}{15}+\frac8{15h(v)}\Big)
      - \frac{1}{2} \log \frac{1+v}{1-v},
\]
where $0<v\le 1$. 
Then $g(0)=0$ and
\[
 g'(v)=\frac{(1-h(v))^3\bigl(3h(v)^2+9h(v)+8\bigr)}{15h(v)^3}\geq0.
\]
Thus $g(v)\geq0$. Dividing by $v>0$ and taking $v=\sqrt x$ proves the
upper bound.

The inequality in \eqref{eq:H3_der} follows from 
\[ L'(x) =\frac{(1-x)^{-1}-L(x)}{2x} 
 \leq\frac1{2(1-x)}-\frac16. \]

\end{proof}

\begin{lemma}
    For all $x\in [0,1)$,
\begin{equation}
     \label{eq:int_bd}\frac{2}{3}\le \int_0^1 \frac{1-\lambda^2}{1-x\lambda^2} \mathrm{d}\lambda \le \frac{2+x}{3}.
 \end{equation}
\end{lemma}

\begin{proof} 
To prove the lower and upper bounds in \eqref{eq:int_bd}, we integrate the inequalities
\[
 1-\lambda^2\leq\frac{1-\lambda^2}{1-x\lambda^2}
 \leq1-(1-x)\lambda^2.
\] 
\end{proof}

 \begin{lemma}
    For $x\ge 1$, we have
    \begin{equation}
    \label{eq:log-est}\log x\leq \frac{x-x^{-1}}{2}.
\end{equation}
\end{lemma}

\begin{proof} 
We have
\[
 \frac{\mathrm{d}}{\mathrm{d} x}\bigl(x-x^{-1}-2\log x\bigr)
 =(1-x^{-1})^2\geq0,
\]
and the expression in parentheses on the left hand side vanishes at $x=1$. 
\end{proof}

 \section{Proof of Proposition \ref{prop:main}} 
 The strategy is to first rewrite $\partial_aF$ in suitable coordinates and factor it into a positive factor times $\Phi$ in \eqref{eq:Phi} below,   reducing   Proposition \ref{prop:main} to proving $\Phi>0$. 
 
 Recall that $a,b>0$, $a+b<1$, and $|t|<1$. 
Define $\gamma_0=\gamma_0(a,b)$,  $y=y(a,b)$, $z=z(a,b)$, $w=w(a,b,t)$, $\eta=\eta(a,b)$ by 
\begin{equation}
    \label{eq:newcoord}
    \gamma_0=\gamma(a,b,0),\qquad  y=\frac{b}{\gamma_0^2+b^2},\qquad z=\frac{ab}{\gamma_0^2},  \qquad w=tz,\qquad  \eta=1-z^2. 
\end{equation}

Since $h(a-b)^2-h(a+b)^2=4ab$,
we obtain
\[
 h(a-b)-h(a+b)=2z\gamma_0, 
\]
and therefore
\begin{equation}
    \label{eq:r}
h(a-b)=\gamma_0(1+z),\qquad h(a+b)=\gamma_0(1-z).
\end{equation}
Taking the product and sum of squares 
 gives 
\begin{equation}
    \label{eq:rprod} h(a-b) h(a+b)=\gamma_0^2\eta,\qquad
 1-a^2-b^2=\gamma_0^2(1+z^2).
\end{equation}

We next derive the bounds on the new coordinates $y,z,w$. We claim that  $0<z<y<1$ and $|w|<z$. Indeed, substituting 
$z$ into the second identity in \eqref{eq:rprod} gives
\[
 (1-a^2-b^2)\gamma_0^2=\gamma_0^4+a^2b^2,
\]
and hence
\begin{equation}\label{eq:expanded-T-product}
 (\gamma_0^2+a^2)(\gamma_0^2+b^2)
 =\gamma_0^4+(a^2+b^2)\gamma_0^2+a^2b^2=\gamma_0^2.
\end{equation}
On the other hand, expanding the definition of $\gamma_0$ and the first equality in \eqref{eq:rprod} gives
\[
 \gamma_0^2=\tfrac{1}{2}(1-a^2-b^2+\gamma_0^2\eta).
\]
Consequently, 
\[
 \gamma_0^2+b^2-b=\tfrac{1}{2}((1-b)^2-a^2+\gamma_0^2\eta)>0,\qquad 
 \gamma_0^2+a^2-a=\tfrac{1}{2}((1-a)^2-b^2+\gamma_0^2\eta)>0.
\]
 where the inqualities follow since $a+b<1$. 
The first inequality implies $0<y<1$. Using
\eqref{eq:expanded-T-product}, the second implies $
 0<z/y<1.$
Finally,  $|w|<z$ follows from $|t|<1$.

Next we derive two identities for $\beta$ and $\gamma$ in terms of the introduced quantities.
\begin{lemma}     \label{lem:Be} Let $\beta= \beta(a,b,t)$ and $\gamma= \gamma(a,b,t)$ be as   in \eqref{eq:r2b} and \eqref{eq:eC}. 
    Then
    \begin{equation}  \label{eq:Be}
 \beta=\frac{2yL(y^2)}b
   +\frac{2ta}{b^2}\Big(yL(y^2)-\frac zaL(z^2)\Big),
 \qquad \gamma=\gamma_0(1+w).
\end{equation}
\end{lemma}

\begin{proof}
Using  \eqref{eq:r}, 
\[\gamma=\tfrac{1}{2}\gamma_0((1+t)(1+z)+(1-t)(1-z))=\gamma_0(1+tz) = \gamma_0(1+w).\]

For $\beta$, 
substituting $x=a+b\lambda$ in the definition of $k$ gives
\[
 k(a,b)=\int_{-1}^1\frac{1-\lambda}{1-(a+b\lambda)^2}\,\mathrm{d}\lambda
 =\frac1{b^2}\int_{a-b}^{a+b}\frac{a+b-x}{1-x^2}\,\mathrm{d} x.
\]
Similarly, 
\[
 k(-a,b)=\frac1{b^2}\int_{a-b}^{a+b}\frac{x-a+b}{1-x^2}\,\mathrm{d} x.
\]
To evaluate these integrals we use 
\begin{equation}
    \label{eq:int1}
     \int_{a-b}^{a+b}\frac{\mathrm{d} x}{1-x^2}
 =\Big[\frac12\log\frac{1+x}{1-x}\Big]_{a-b}^{a+b}
 =2yL(y^2),
\end{equation}
\begin{equation}
    \label{eq:int2}
     \int_{a-b}^{a+b}\frac{x\,\mathrm{d} x}{1-x^2}
 =\Big[-\frac12\log(1-x^2)\Big]_{a-b}^{a+b}
 =\log\frac{h(a-b)}{h(a+b)}=2zL(z^2), 
\end{equation}
which gives
\[
 k(a,b)=\frac{2(a+b)yL(y^2)-2zL(z^2)}{b^2},\qquad
 k(-a,b)=\frac{2(b-a)yL(y^2)+2zL(z^2)}{b^2}.
\]
Substituting these in the definition of  $\beta$ yields the claim. 

To see that the logarithms in \eqref{eq:int1} and \eqref{eq:int2} evaluate to  $2yL(y^2)$ and $2zL(z^2)$, respectively, 
we first observe that by  \eqref{eq:H1_expr}, 
\[
 2zL(z^2) =\log\frac{1+z}{1-z} = \log\frac{h(a-b)}{h(a+b)}.
\]
Further, we have
\[
 2yL(y^2)=\log\frac{1+y}{1-y}
 =\frac12\log\frac{(1+a+b)(1-a+b)}{(1+a-b)(1-a-b)}.
\]
This can be seen by first expanding 
\eqref{eq:expanded-T-product} to obtain
\[
 (\gamma_0^2+b^2)^2+b^2=(1-a^2+b^2)(\gamma_0^2+b^2).
\]
Adding or subtracting $2b(\gamma_0^2+b^2)$ on the two sides gives
\[
 (\gamma_0^2+b^2\pm b)^2
 =(\gamma_0^2+b^2)[(1\pm b)^2-a^2].
\]
Using the definition of $y$, we therefore have
\[
 \Big(\frac{1+y}{1-y}\Big)^2
 =\frac{(\gamma_0^2+b^2+b)^2}{(\gamma_0^2+b^2-b)^2}
 =\frac{(1+b)^2-a^2}{(1-b)^2-a^2},
\]
 which gives the desired formula for $2yL(y^2)$.
\end{proof}

We derive another identity for $\beta$. Let
$\theta=\theta(a,b,t)$ be given by
\begin{equation}
    \label{eq:E}
     \theta=\eta\int_0^1
 \frac{1-\lambda^2}{(1-{z^2}\lambda^2)(1-{y^2}\lambda^2)}\,\mathrm{d}\lambda
 =\eta\frac{\eta L(z^2)-(1-y^2)L(y^2)}{y^2-z^2}.
\end{equation}
\begin{lemma}
    \label{lem:B}
    The following identity holds:
   \begin{equation} \label{eq:B}
       \frac{ b }{2y} \beta = L(y^2) - \frac{w \theta}{\eta}.
   \end{equation} 
\end{lemma}
\begin{proof}
By  the definition of $y$,  
\[
 (\gamma_0^2+b^2)(1-y^2)
 =\gamma_0^2+b^2-\frac{b^2}{\gamma_0^2+b^2}
 =\gamma_0^2-\frac{a^2b^2}{\gamma_0^2}=\gamma_0^2 \eta.
\]
It follows that
\[
 1-y^2=\frac{\gamma_0^2\eta}{\gamma_0^2+b^2},\qquad
 y^2-z^2=\eta-(1-y^2)=\frac{b^2\eta}{\gamma_0^2+b^2}.
\]
Substituting these identities into the right-hand side of \eqref{eq:E} gives
\[
\theta =\frac \eta{b^2}
 \bigl((\gamma_0^2+b^2)L(z^2)-\gamma_0^2L(y^2)\bigr).
\]
Multiplying both sides of the identity for $\beta$ in \eqref{eq:Be} by $(\gamma_0^2+b^2)/2$ yields
\[
\frac{\gamma_0^2+b^2}{2} \beta = (\gamma_0^2+b^2) \Big( \frac{yL(y^2)}b
   +\frac{ta}{b^2}\Big(yL(y^2)-\frac zaL(z^2)\Big) \Big).
\]
By the definition of $y$, $z$, $w$, and the identity for $\theta$, we obtain the chain of inequalities with
\[
 \frac{b}{2y} \beta = L(y^2)
   + \frac{ta}{b}\Big(L(y^2)- \frac{\gamma_0^2+b^2}{\gamma_0^2} L(z^2)\Big) = L(y^2)
   - \frac{w \theta}{\eta} .
\]
\end{proof}

We will use  $u=ta$. 
Differentiation  yields $\frac{d}{da}\gamma_0(a,b)=-(a(\gamma_0^2+b^2))/{\gamma_0^3\eta}$. Together with the identity for $\gamma$ in \eqref{eq:Be} and for $\beta$ in \eqref{eq:B}, this yields 
\[ \frac{d}{da}(\gamma(a,b,ua^{-1}))=-\frac{a(\gamma_0^2+b^2)}{\gamma_0^3\eta}(1-w)
,\qquad \frac{d}{da}(\beta(a,b,ua^{-1}))=\frac{4a}{\gamma_0^4\eta^2}(1-w\mu),
\]
where $ \mu= \mu(a,b,t)$ is given by
\begin{equation}
    \label{eq:M}
     \mu=1+\frac {\eta}{2}\int_0^1\frac{1-\lambda^2}{1-z^2\lambda^2}\,\mathrm{d}\lambda
   =\frac{1+z^2-\eta^2L(z^2)}{2z^2}.
\end{equation}
Also, $\frac{d}{da}(\delta(a,b,ua^{-1}))=\frac{d}{da}(\gamma(a,b,ua^{-1}))$.

By the product rule and formulas above, 
\[
\partial_a G(a,b,u) = \frac{d}{da}((\delta^2\beta)(a,b,ua^{-1}))=2\delta\beta \frac{d}{da}(\gamma(a,b,ua^{-1}))+\delta^2\frac{d}{da}(\beta(a,b,ua^{-1})) \]
\[=\frac{4a\delta^2}{\gamma_0^4\eta^2}(1-w\mu)
 -\frac{2a\delta\beta (\gamma_0^2+b^2)}{\gamma_0^3\eta}(1-w).
\]
Using $\gamma_0^2+b^2=b/y$ and 
 Lemma \ref{lem:B} gives 
 \[\frac{d}{da}((\delta^2\beta)(a,b,ua^{-1})) = \frac{4a\delta}{\gamma_0^3\eta^2}\Phi,\]
where  $ \Phi= \Phi(a,b,t)$ is 
\begin{equation}\label{eq:Phi}
\Phi=\frac{\delta}{\gamma_0}(1-w\mu)-(1-w)\eta\Big( L(y^2)-\frac{w\theta}{\eta}\Big).
\end{equation}
Since the factor in front of $\Phi$   is positive, it suffices to prove that $\Phi$ is positive.

\subsection{Lower bounds for \texorpdfstring{$\Phi$}{positivity}} 
Recall the definitions of $\mu$ in \eqref{eq:M} and $\theta$ in \eqref{eq:E}. 
 Using \eqref{eq:int_bd} and 
$\eta/(1-z^2\lambda^2)\leq1$ in the integral defining $\theta$ gives
\begin{equation}\label{eq:MEbounds}
 1+\frac \eta 3\leq \mu\leq1+\frac{\eta(2+z^2)}6,
 \qquad
 \frac{2\eta}{3}\leq \theta\leq\frac{2+y^2}{3}.
\end{equation}

Since $|w|<z<1$, we have $1-w>0$.   We split our analysis at $w=0$ because in the expression for $\Phi$, the coefficients of
$\mu$ and $\theta$  change sign there. Denote 
\[\Gamma_-=\frac{\delta}{\gamma_0}\Big(1-w\Big(1+\frac \eta 3\Big)\Big)-(1-w)\Big(
 \eta\cdot  \Big(\frac{7+y^2}{15}+\frac8{15\sqrt{1-y^2}}\Big)
 -w\frac{2+y^2}{3} \Big),\]
\[\Gamma_+= \frac{\delta}{\gamma_0}\Big(1-w\Big(1+\frac{\eta(2+z^2)}6\Big)\Big)
 -(1-w)\Big (
 \eta \cdot \Big(\frac{7+y^2}{15}+\frac8{15\sqrt{1-y^2}}\Big)
 -\frac{2w\eta}{3}\Big). \]
 When $w\leq0$, equations
\eqref{eq:H2_bounds} and \eqref{eq:MEbounds} give 
\[\Phi\geq \Gamma_-,\]
while for $w\geq0$,
\[\Phi\geq \Gamma_+.  \]

We will also use a sharper lower bound in the region 
\begin{equation}\label{region-bounds}
    y\geq\frac45,\qquad  \frac zy\geq\frac45,\qquad  w\geq z^2.
\end{equation}
We rewrite 
\[
 \eta L(y^2)-w\theta =\eta(1-w)L(z^2)
 +\eta\Big(1+\frac{w(1-y^2)}{y^2-z^2}\Big)\bigl(L(y^2)-L(z^2)\bigr).
\]
Together with the formula for $\mu$, this gives
\[
 \Phi=\frac{\delta}{\gamma_0}\Big(1-\frac{w(1+z^2)}{2z^2}\Big)
 +\eta\cdot \Big(\frac{w\eta\delta }{2z^2\gamma_0}-(1-w)^2\Big)L(z^2) \]
 \[-(1-w)\eta\cdot \Big(1+\frac{w(1-y^2)}{y^2-z^2}\Big)
 (L(y^2)-L(z^2)).
\]

The coefficient of $L(z^2)$ is positive since 
\begin{equation}\label{eq:positive-HZ-coefficient}
 \eta\cdot \Big(\frac{w\eta \delta}{2z^2\gamma_0}-(1-w)^2\Big)
 \geq\frac{\eta^2(3z^2-1)}2>0.
\end{equation}
by the definition of the region \eqref{region-bounds}, $z^2\geq256/625>1/3$. Also,
$1-w\leq \eta$ and $\delta/\gamma_0\geq1+w\geq1+z^2$.)

On the other hand, the coefficient of 
$L(y^2)-L(z^2)$ is negative. 
Using \eqref{eq:log-est}   applied to  
$\sqrt{\eta/(1-y^2)}$, we obtain 
\[
 \frac12\log\frac \eta{1-y^2}
 \leq\frac{y^2-z^2}{2\sqrt{\eta(1-y^2)}}
 \leq\frac{y(y^2-z^2)}{2\sqrt{(y^2-z^2)(1-y^2)}}.
\]
Consequently, using also  \eqref{eq:H3_der}   with $y^2$ and $z^2$, 
\begin{equation}\label{eq:diffbound}
 L(y^2)-L(z^2)\leq\frac12\log\frac \eta{1-y^2}-\frac{y^2-z^2}{6} \leq
 (y^2-z^2)\Big(\frac{y}{2\sqrt{(y^2-z^2)(1-y^2)}}-\frac16\Big).
\end{equation}
Together with the estimate 
$L(z^2)\geq1+z^2/3$  this  gives 
\[\Phi \geq\Gamma_{++},\] where
\[
\Gamma_{++}=\frac{\delta}{\gamma_0}\Big(1-w\Big(1+\frac{\eta(2+z^2)}6\Big)\Big) -\eta(1-w)^2\Big(1+\frac {z^2}3\Big)  \]
      \begin{equation}\label{eq:Psi}   -(1-w)\eta\bigl(y^2-z^2+w(1-y^2)\bigr)
 \Big(\frac{y}{2\sqrt{(y^2-z^2)(1-y^2)}}-\frac16\Big).
\end{equation}

To complete the proof of the theorem  it remains to establish   $\Gamma_->0$ when $w\le 0$,  $\Gamma_{++}>0$ when  \eqref{region-bounds} holds, and $\Gamma_+>0$ in the remainig region where $w\ge 0$. 
\subsection{Change of variables}   We  
change variables again, to express $\Gamma_-,\Gamma_+,\Gamma_{++}$ using rational functions
and one remaining square root.  
Recall that $0<z<y<1$ and $w=tz$. 
First introduce 
 \begin{equation}\label{eq:roots-rhosigma}
 v=\frac{z}{y}.
\end{equation}

 We express $a,b,\gamma_0$ in terms of $y,v$. The identity \eqref{eq:expanded-T-product}  gives
$v=a/(\gamma_0^2+a^2)$ and hence
\[
 v\pm y
 =\frac{a(\gamma_0^2+b^2)\pm b(\gamma_0^2+a^2)}{\gamma_0^2}
 =(a\pm b)(1\pm z).
\]
Isolating $a\pm b$ and adding and subtracting the identites for $a+b,a-b$ gives 
\begin{equation}\label{eq:roots-abT}
 a=\frac{v(1-y^2)}\eta,
 \qquad b=\frac{y(1-v^2)}\eta=\frac{y(1-v^2)}\eta.
\end{equation}
Finally, $\gamma_0={h(y) h(v)}/\eta$ because $\gamma_0^2=ab/z$.

We begin with $\Gamma_-$. 
By definition, 
\[\Gamma_- = \frac{\delta}{\gamma_0}\mathcal{G}_- - (1-w)\Big(\eta\Big(\frac{7+y^2}{15}+\frac{8}{15h(y)}\Big)-\frac{(2+y^2)\,w}{3}\Big),\qquad \mathcal{G}_-=1-w\Big(1+\frac \eta 3\Big). \]
By \eqref{eq:roots-abT}, $ u-b=ta-b=(tv ( 1- y^2)-y(1-v^2))/\eta$. Also, 
  $ |u-b|\leq|t|a+b<a+b<1 $. Hence, with 
\begin{equation}\label{eq:Q-minus}
 \mathcal Q=\eta^2-(tv(1-y^2)-y(1-v^2))^2 =\eta^2\bigl(1-(u-b)^2\bigr)>0,
\end{equation}
we have $h(u-b)=\sqrt{\mathcal Q}/\eta$, and using $\gamma_0=h(y)h(v)/\eta$,
\[
\frac{\delta}{\gamma_0}=1+w+\frac{1-h(u-b)}{\gamma_0}=1+w+\frac{\eta-\sqrt{\mathcal Q}}{h(y)h(v)},
\]
Multiplying $\Gamma_-$ by $h(y)h(v)$ we obtain
\begin{equation}\label{eq:Phi-minus-identity}
 h(y)h(v)\,\Gamma_-=\mathcal F_--\mathcal{G}_-\sqrt{\mathcal Q},
\end{equation}
with
\[
\mathcal F_-=\mathcal{G}_- (\eta+(1+w)h(y)h(v))
\]
\begin{equation}\label{eq:F-minus}
 -(1-w)\Big(\frac{\eta h(v)(h(y)(7+y^2)+8)}{15}-\frac{(2+y^2)\,w\,h(y)h(v)}{3}\Big).
\end{equation}

For $\Gamma_+$, we similarly obtain 
\[ h(y)h(v)\,\Gamma_+=\mathcal F_+-\mathcal{G}_+\sqrt{\mathcal Q},\]
with $\mathcal{Q}$ as  in \eqref{eq:Q-minus}, 
\[
 \mathcal{G}_+=1-w\Big(1+\frac{\eta(2+z^2)}{6}\Big),
\]
\[
 \mathcal F_+=\mathcal{G}_+(\eta+(1+w)h(y)h(v))
 -(1-w)\Big(\frac{\eta h(v)(h(y)(7+y^2)+8) }{15}-\frac{2w \eta\,h(y)h(v)}{3}\Big).
\]

In case of  $\Gamma_{++}$, we obtain 
\[ h(y)h(v)\,\Gamma_{++}=\mathcal F_{++}-\mathcal{G}_+\sqrt{\mathcal Q},\]
where $\mathcal{G}_+,\mathcal{Q}$ are as before, and 
\[
 \mathcal F_{++}=\mathcal{G}_+(\eta+(1+w)h(y)h(v))
 -\eta h(y)h(v)(1-w)^2\Big(1+\frac{z^2}{3}\Big)
 \]
 \[
 -(1-w)\eta\bigl(y^2(1-v^2)+w(1-y^2)\bigr)\Big(\frac12-\frac{h(y)h(v)}{6}\Big).
\]

To show positivity of $\Gamma_-,\Gamma_+,\Gamma_{++}$, it suffices to show
\begin{equation}
    \label{eq:oldcoord-toshow}
    \mathcal F_--\mathcal{G}_-\sqrt{\mathcal Q}>0, \qquad \mathcal F_+-\mathcal{G}_+\sqrt{\mathcal Q}>0, \qquad \mathcal F_{++}-\mathcal{G}_+\sqrt{\mathcal Q}>0,
\end{equation}
in the respective regions. 

We now introduce new variables $j,k\in(0,1)$ by  
\begin{equation}\label{eq:roots-rationalization}
 v=\frac zy=\frac{2j}{1+j^2},\qquad
 y=\frac{2k}{1+k^2}.
\end{equation}
Note that the map $x\mapsto 2x/(1+x^2)$ is a bijection from $(0,1)$ to itself.
The reason for choosing this particular parametrization is that
\[
 1-v^2=\frac{(1-j^2)^2}{(1+j^2)^2},\qquad
 1-y^2=\frac{(1-k^2)^2}{(1+k^2)^2},
\]
so their square roots $h(v),h(y)$  are    rational  in new
coordinates. We have 
\[z=vy=\frac{4jk}{(1+j^2)(1+k^2)}.\]
 
Now, any
$0<j,k<1$ gives $0<v,y<1$, and these formulas give $a,b>0$ with
\[
 1-(a+b)=\frac{(1-v)(1-y)}{1+vy}>0.
\]
Thus the change of coordinates from $\{(j,k) \in (0,1)^2 \}$ to $\{ (a,b) \in (0,1)^2 \colon a+b < 1 \}$ is bijective.

In the variables $j,k,t$, the functions $\mathcal Q, \mathcal{G}_-,\mathcal F_-, \mathcal{G}_+,\mathcal F_+, \mathcal{G}_{++},\mathcal F_{++}$ are polynomials in $t$ with coefficients that are rational in $j,k$.  Consequently,
$\mathcal F_-^2-\mathcal{G}_-^2\mathcal Q_-$ has no square
roots and can be turned into a polynomial by clearing positive
denominators. This will be used in the next section. 

\subsection{Polynomial inequalities}
By abuse of notation, we will denote 
\[
\mathcal F_- = \mathcal F_-(j,k,t)
\]
to mean  \eqref{eq:F-minus} evaluated at $v=2j/(1+j^2),\,y=2k/(1+k^2)$, and likewise for $\mathcal G_\pm(j,k,t)$, $\mathcal Q(j,k,t)$, $\mathcal F_+(j,k,t)$, $\mathcal F_{++}(j,k,t)$,  $z(j,k)$, and $\eta(j,k)$. 

Now, to show \eqref{eq:oldcoord-toshow}, it suffices to show that 
\begin{equation}\label{eq:toshow-new}
 \mathcal F^2-\mathcal G^2\mathcal Q>0\qquad\text{and}\qquad\mathcal F>0
\end{equation}
holds at every point $(j,k,t)$, for $(\mathcal F,\mathcal G)=(\mathcal F_-,\mathcal G_-)$ when $-1<t\leq0$; for $(\mathcal F,\mathcal G)=(\mathcal F_+,\mathcal G_+)$ when $0\leq t<1$ and $(j,k,t)$ lies outside the region \eqref{region-bounds}; and for $(\mathcal F,\mathcal G)=(\mathcal F_{++},\mathcal G_+)$ when $0\leq t<1$ and $(j,k,t)$ lies in that region. In the coordinates $(j,k,t)$ the region \eqref{region-bounds} reads $j,k\geq\frac12$, $t\geq z(j,k)$. Indeed, these inequalities give
\[
 \mathcal F>|\mathcal G|\sqrt{\mathcal Q}\geq\mathcal G\sqrt{\mathcal Q},
\]
hence $h(y)h(v)\Gamma_->0$, $h(y)h(v)\Gamma_+>0$, and $h(y)h(v)\Gamma_{++}>0$, respectively.

 We now turn the squared inequalities \eqref{eq:toshow-new} into polynomial inequalities on $[0,1]^2\times [-1,0]$ and $[0,1]^2\times [0,1]$, respectively. The signs of $\mathcal F_-,\mathcal F_+,\mathcal F_{++}$ will be established in the end. 
Let \[
 \mathcal U_-=(0,1)^2\times[-1,0],\qquad  \mathcal U_+=(0,1)^2\times[0,1].\]
Thus, the variables $(j,k,t)$ range over
$\mathcal U_-$
 in the case of $\Gamma_-$ and over 
$\mathcal U_+$
 in the case of $\Gamma_+$ and $\Gamma_{++}$. The argument below covers the  boundaries $\pm 1$, although we don't need that.

Set $d=(1+j^2)(1+k^2)$  and define the positive denominators
\[
 d_{-}=15(1+k^2)d^5,\qquad
 d_+=15d^7,\qquad
 d_{++}=3d^7,
\]
Substitution of \eqref{eq:roots-rationalization} and
\eqref{eq:roots-rhosigma} show  that 
\[P_- = d_-^2
 (\mathcal F_-^2-\mathcal G_-^2\mathcal Q), \qquad  P_+ =\frac{d_+^2
 (\mathcal F_+^2-\mathcal G_+^2\mathcal Q)}{d^2-16j^2k^2}, \qquad P_{++}=\frac{d_{++}^2
 (\mathcal F_{++}^2-\mathcal G_+^2\mathcal Q)}{d^2-16j^2k^2}, \]
 are
polynomials in $j,k,t$ with integer coefficients. 
Note that $z={4jk}/{d}$.

Since $d^2-16j^2k^2=(d-4jk)(d+4jk)>0$ on $\mathcal{U}$,  division by it preserves the sign.
Thus, to establish the first inequality \eqref{eq:toshow-new},  it suffices to show that $P_->0$ on $\mathcal U_-$ and that $P_+,P_{++}>0$ on the respective subsets of $\mathcal U_+$ listed below.

For this we will use the   substitutions in
Table~\ref{tab:seven-substitutions}.
For     $\nu\in
\{A1,A2,\dots,C1\}$, let 
\[
 \chi_\nu:[0,1]^3\longrightarrow[0,1]^2\times[-1,1],
 \qquad x=(x_1,x_2,x_3)\mapsto(j,k,t),
\]
be the substitution indicated in columns $4-6$ of  Table~\ref{tab:seven-substitutions}. For instance,  $\chi_{A_1}(x)=(x_1x_2,x_2,-x_3)$, i.e. we substitute $j=x_1x_2, \, k=x_2,\, t=-x_3$.    
Set
\[
 \Omega_\nu
 =\{x\in[0,1]^3:\chi_\nu(x)\in\mathcal U\},
 \qquad
 \mathcal U_\nu=\chi_\nu(\Omega_\nu).
\]
 Each
restricted map $\chi_\nu:\Omega_\nu\to\mathcal U_\nu$
is a bijection and the regions  cover $\mathcal{U}_{\pm}$: 
\[
 \mathcal U_-
 =\mathcal U_{\mathrm{A1}}
  \cup\mathcal U_{\mathrm{A2}}
  \cup\mathcal U_{\mathrm{A3}},\qquad 
 \mathcal U_+
 =\mathcal U_{\mathrm{B1}}
  \cup\mathcal U_{\mathrm{B2}}
  \cup\mathcal U_{\mathrm{B3}}
  \cup\mathcal U_{\mathrm{C1}}.
\]
Thus, $\mathcal U_\nu$ is a region in the original coordinates $(j,k,t)$ 
(its defining conditions are listed in the third column), while $\Omega_\nu$ is a set of new parameters (last column).  In the $A$-rows, $t\leq0$, so that $-t=|t|\in[0,1]$.
The substitutions in     Table \ref{tab:seven-substitutions}  record ratios within the defining inequalities of $\mathcal{U}_\nu$.  In the row C1 we write \[\zeta(x)=\zeta_{\mathrm{1}}(x_1,x_2)+(1-\zeta_{\mathrm{1}}(x_1,x_2))\,x_3,\] with $\zeta_{\mathrm{3}}(x_1,x_2)=4x_1x_2(1+x_1^2)^{-1}(1+x_2^2)^{-1}$, which is $z$ in the new coordinates, and $\zeta_{\mathrm{1}}(x_1,x_2)=16(1+x_1)(1+x_2)(4+(1+x_1)^2)^{-1}(4+(1+x_2)^2)^{-1}$, which is $z$ at $(1+x_1)/2,(1+x_2)/2$. 

 \begin{table}[ht]
\centering
\small
\begin{tabular}{lllllll}
\hline
$\nu$ & Pol. &  $\mathcal U_\nu$ & $j$ & $k$ & $t$ &   $\Omega_\nu$\\
\hline
A1 & $P_{-}$ & $j\leq k$ & $x_1x_2$ & $x_2$ & $-x_3$ & $(0,1]\times(0,1)\times[0,1]$\\
A2 & $P_{-}$ & $k\leq -jt$ & $x_1$ & $x_1x_2x_3$ & $-x_3$ & $(0,1)\times(0,1]\times(0,1]$\\
A3 & $P_{-}$ & $-jt\leq k\leq j$ & $x_1$ & $x_1x_2$ & $-x_2x_3$ & $(0,1)\times(0,1]\times[0,1]$\\
B1 & $P_+$ & $j\leq 1/2$ & $\tfrac{1}{2}x_1$ & $x_2$ & $x_3$ & $(0,1]\times(0,1)\times[0,1]$\\
B2 & $P_+$ & $j\geq 1/2,\ k\leq 1/2$ & $\tfrac{1}{2}(1+x_1)$ & $\tfrac{1}{2}x_2$ & $x_3$ & $[0,1)\times(0,1]\times[0,1]$\\
B3 & $P_+$ & $t\leq z$ & $x_1$ & $x_2$ & $\zeta_{\mathrm{3}}(x_1,x_2)\,x_3$ & $(0,1)\times(0,1)\times[0,1]$\\
C1 & $P_{++}$ & $j,k\geq 1/2,\ t\geq z$ & $\tfrac{1}{2}(1+x_1)$ & $\tfrac{1}{2}(1+x_2)$
   & $\zeta(x)$ & $[0,1)\times[0,1)\times[0,1]$ \\
\hline
\end{tabular}
\label{tab:seven-substitutions} \medskip \caption{Column 2: the polynomial in $(j,k,t)$ whose positivity is proved on $\mathcal U_\nu$. Column 3: the region $\mathcal U_\nu$ in $(j,k,t)$. Columns 4--6: the substitution $(j,k,t)=\chi_\nu(x_1,x_2,x_3)$. Column 7:  the region $\Omega_\nu\subseteq[0,1]^3$ in $(x_1,x_2,x_3)$.}
\end{table}

For instance, in $A1$, the polynomial in new coordinates  is 
\[
 P_{-1}(x_1x_2,x_2,-x_3).
\]
After the substitutions, we divide out a positive factor wherever possible: in $A1$ we divide by $x_2^2$, in $A2$ by  $x_1^2x_3^2$, in $A3$ by $x_1^2x_2^2$, in $B3$ by $x_2^2$. 
We  obtain
\begin{align*}
 Q_{\mathrm{A1}}(x)&=x_2^{-2}\,P_-(x_1x_2,\,x_2,\,-x_3),\\
 Q_{\mathrm{A2}}(x)&=(x_1^2x_3^2)^{-1}\,P_-(x_1,\,x_1x_2x_3,\,-x_3),\\
 Q_{\mathrm{A3}}(x)&=(x_1^2x_2^2)^{-1}\,P_-(x_1,\,x_1x_2,\,-x_2x_3),\\
 Q_{\mathrm{B1}}(x)&=2^{24}\,P_+(\tfrac{x_1}{2},\,x_2,\,x_3),\\
 Q_{\mathrm{B2}}(x)&=2^{48}\,P_+(\tfrac{1+x_1}{2},\,\tfrac{x_2}{2},\,x_3),\\
 Q_{\mathrm{B3}}(x)&=x_2^{-2}\,R_+(x_1,x_2,x_3),\\
 Q_{\mathrm{C1}}(x)&=2^{56}\,R_{++}(\tfrac{1+x_1}{2},\,\tfrac{1+x_2}{2},\,x_3),
\end{align*}
where $R_+$ and $R_{++}$ are the polynomials in $(j,k,t)$ obtained from $P_+$, $P_{++}$ by rescaling the last variable to the interval $[0,z(j,k)]$, respectively $[z(j,k),1]$, and clearing denominators:
\[
 R_+(j,k,t)=d^4\,P_+\Big(j,k,\frac{4jk}{d}\,t\Big),
 \]
 \[
 R_{++}(j,k,t)=\frac{d^4}{(d-4jk)^2}\,P_{++}\Big(j,k,\frac{4jk+(d-4jk)\,t}{d}\Big).
\]
It now suffices to prove 
 $Q_\nu(x)>0$ for all $x\in \Omega_\nu$, for each $\nu$.

\subsection{Positivity of the seven polynomials}
We now verify positivity of the seven polynomials $Q_\nu$ constructed
in the previous section. We express 
each polynomial in  Bernstein
basis, which uses products of powers of $x_i$ and $1-x_i$, all
of which are non-negative on $[0,1]^3$. 

If $Q(x)=\sum_\alpha q_\alpha x^\alpha$ is a polynomial with integer
coefficients (here $\alpha=(\alpha_1,\alpha_2,\alpha_3)$,
$x^\alpha=x_1^{\alpha_1}x_2^{\alpha_2}x_3^{\alpha_3}$),  and  if  $n_i=\deg_{x_i}Q$,   we can write  
\begin{equation}\label{eq:bernstein-expansion}
 Q(x)=\sum_{\beta\leq n}b_\beta
 \prod_{i=1}^3x_i^{\beta_i}(1-x_i)^{n_i-\beta_i}.
\end{equation}
Here  
$\alpha\leq\beta$ means $\alpha_i\leq\beta_i$ for each $i$ and  for $0\leq\beta_i\leq n_i$,  
\begin{equation}\label{eq:bernstein-coefficients}
 b_\beta=\sum_{\alpha\leq\beta}q_\alpha
 \prod_{i=1}^3\binom{n_i-\alpha_i}{\beta_i-\alpha_i}.
\end{equation}

To see this, for one variable the binomial formula gives
\[
 x_i^{\alpha_i}
 =x_i^{\alpha_i}[x_i+(1-x_i)]^{n_i-\alpha_i}
 =\sum_{\beta_i=\alpha_i}^{n_i}
 \binom{n_i-\alpha_i}{\beta_i-\alpha_i}
 x_i^{\beta_i}(1-x_i)^{n_i-\beta_i}.
\]
Multiply these identities for $i=1,2,3$, multiply by $q_\alpha$,
and sum over $\alpha$. Collecting the coefficient   gives \eqref{eq:bernstein-coefficients}. 

Since  $x_i$ and  $1-x_i$ are non-negative on $[0,1]^3$, to   obtain $Q\geq0$ on $[0,1]^3$ it suffices to show 
$b_\beta\geq 0$ for every $\beta$.  
Furthermore, to conclude that $Q(x)>0$  at a point $x \in \Omega_{\nu}$,  it suffices to find one index
$\beta$ such that
\begin{equation}
    \label{pos_toshow}
     b_\beta>0,\qquad
 \prod_{i=1}^3x_i^{\beta_i}(1-x_i)^{n_i-\beta_i}>0.
\end{equation}
Such an index will be called a witness
for strict positivity at $x$.

We apply this to the seven polynomials $Q_\nu$. At an interior point of the cube every basis element in
\eqref{eq:bernstein-expansion} is strictly positive, so any positive
coefficient $b_\beta$ gives $Q_\nu(x)>0$. On a face $x_i\in\{0,1\}$,
however, the factors $x_i$ or $1-x_i$   vanish, so strict positivity there requires a positive
coefficient whose index has $\beta_i=0$ or $\beta_i=n_i$. 

Computing the coefficients
\eqref{eq:bernstein-coefficients}, we see 
that none of them is negative, and for each $\nu$ one can find one or two
indices $\beta$ with $b_\beta>0$  so that 
the product in \eqref{pos_toshow} is  strictly positive at every
point of $\Omega_\nu$. See 
Table~\ref{tab:bernstein-witnesses}, which lists for each $Q_\nu$ its degrees
$(n_1,n_2,n_3)$, the numbers of strictly positive and of vanishing
coefficients $b_\beta$, and the selected witness indices. For each row the two
counts sum to $(n_1+1)(n_2+1)(n_3+1)$, the number of all indices $\beta$ in
\eqref{eq:bernstein-expansion}. Altogether,  $52,335$ coefficients must be checked. 
This check can be performed easily with a computer program. Each of these computations has been formally verified as part of the Lean formalization.
Table~\ref{tab:witness-values} gives the exact integer
 values of the coefficients $b_\beta$ at the thirteen witness indices. 
\begin{table}[ht] 
\centering
\begin{tabular}{llrrl}
\hline
$\nu$ & $(n_1,n_2,n_3)$ & $b_\beta>0$ & $b_\beta=0$ & Witness $(\beta_1,\beta_2,\beta_3)$\\
\hline
A1 & $(20,42,4)$ & 4465 & 50 & $(20,1,0)$, $(20,1,4)$\\
A2 & $(42,24,24)$ & 26855 & 20 & $(1,24,24)$\\
A3 & $(42,24,4)$ & 5283 & 92 & $(1,24,0)$, $(1,24,4)$\\
B1 & $(24,24,4)$ & 3040 & 85 & $(24,2,0)$, $(24,2,4)$\\
B2 & $(24,24,4)$ & 3003 & 122 & $(0,24,0)$, $(0,24,4)$\\
B3 & $(32,30,4)$ & 5017 & 98 & $(1,1,0)$, $(1,1,4)$\\
C1 & $(28,28,4)$ & 4165 & 40 & $(0,0,0)$, $(0,0,4)$\\
\hline
\end{tabular}\medskip
\caption{The coefficient counts and thirteen positive witnesses.}
\label{tab:bernstein-witnesses}
\end{table}

To elaborate on Table \ref{tab:bernstein-witnesses}, an endpoint index $\beta_i=n_i$ is
used where the face $x_i=1$ lies in $\Omega_\nu$, so that the factor $1-x_i$
does not appear: $\beta_1=20$ in A1 (corresponding to $j=k$), $\beta_2=24$ in
A3 ($k=j$), $\beta_1=24$ in B1 ($j=1/2$), and $\beta_2=24$ in B2
($k=1/2$). An endpoint index $\beta_i=0$ is used where the face $x_i=0$
lies in $\Omega_\nu$, so that the factor $x_i$ does not appear: $\beta_1=0$ in
B2 ($j=1/2$) and $\beta_1=\beta_2=0$ in C1 ($j=k=1/2$). Where $x_i\in (0,1)$, any index is
admissible, and a small one is taken. In the third coordinate, $x_3$ ranges
over $[0,1]$, so both faces $x_3=0$ and $x_3=1$ belong to $\Omega_\nu$, and
each row lists the pair $\beta_3=0$ and $\beta_3=4$, whose basis factors
$(1-x_3)^4$ and $x_3^4$ cannot vanish simultaneously. The exception is A2,
where $k=x_1x_2x_3>0$ excludes $x_3=0$, so the single index
$\beta=(1,24,24)$ with $\beta_2=n_2$ and $\beta_3=n_3$ covers the admissible
faces $x_2=1$ and $x_3=1$.  

\begin{table}[ht]
\centering
\begin{tabular}{lrr}
\hline
$\nu$ & First witness coefficient & Second witness coefficient\\
\hline
A1 & $12600$ & $25200$\\
A2 & $25200$ & ---\\
A3 & $12600$ & $25200$\\
B1 & $109335937500$ & $8254492187500$\\
B2 & $3658583003906250000$ & $325224219836275200$\\
B3 & $288000$ & $288000$\\
C1 & $338787737221858106544$ & $62644543774312500000$\\
\hline
\end{tabular} \medskip
\caption{Values of the coefficients $b_\beta$ at the witness indices of Table~\ref{tab:bernstein-witnesses}.}
\label{tab:witness-values}
\end{table}

\subsection{\texorpdfstring{Positivity of $\mathcal{F}_-$, $\mathcal{F}_+$, $\mathcal{F}_{++}$}{Remaining positivity arguments}}
It remains to show the second inequality in \eqref{eq:toshow-new} for each $\mathcal{F}\in\{\mathcal{F}_-,\mathcal{F}_+,\mathcal{F}_{++}\}$.
From the first inequality in \eqref{eq:toshow-new}, and $Q_{\nu}>0$  we already know 
$\mathcal{F}^2  > 0$ for each such choice of $\mathcal{F}$. 

Let
$c_0=(1/2,1/2,1/2)$. 
 We will use continuity to determine
the sign from evaluation of the expressions in questions at $\chi_\nu(c_0)$. For instance, $\chi_{A1}(c_0)=(1/4,1/2,-1/2)$ and 
\[
 \mathcal F_{-}(1/4,1/2,-1/2)
 =\frac{16433260569}{22185265625}>\frac12.
\]
These lower bounds are listed in Table \ref{tab:positive-centers}.
\begin{table}[ht]
\centering
\begin{tabular}{lll}
\hline
$\nu$ & $\chi_\nu(c_0)$ &  Lower bound\\
\hline
A1 & $(1/4,1/2,-1/2)$ & $\mathcal F_{-}>1/2$\\
A2 & $(1/2,1/8,-1/2)$ & $\mathcal F_{-}>1/2$\\
A3 & $(1/2,1/4,-1/4)$ & $\mathcal F_{-}>1/2$\\
B1 & $(1/4,1/2,1/2)$ & $\mathcal F_+>1/2$\\
B2 & $(3/4,1/4,1/2)$ & $\mathcal F_+>1/2$\\
B3 & $(1/2,1/2,8/25)$ & $\mathcal F_+>1/2$\\
C1 & $(3/4,3/4,1201/1250)$ & $\mathcal F_{++} >1/100$\\
\hline
\end{tabular}\medskip
\caption{Seven   evaluations.}
\label{tab:positive-centers}
\end{table}

We now justify that these evaluations determine the signs throughout
the required regions, including the boundaries. Fix $\nu$
and $x\in\Omega_\nu$, 
$\chi_\nu(x)=(j,k,t)$. Join $x$ to $c_0$ by the line segment 
\[
 \ell(s)=(1-s)x+sc_0,\qquad 0\leq s\leq1.
\]
This segment lies in $\Omega_\nu$ for $0\leq s\leq1$,
and in $(0,1)^3$ for $0<s\leq1$.
Inspection of each substitution
shows that $\chi_\nu(\ell(s))$ then has $0<j,k<1$ and remains in
the same region.  In particular,  \eqref{eq:toshow-new}  holds on  this line segment, including its initial point.

Define
\[
f(s) = \begin{cases}
    \mathcal F_{-}(\chi_\nu(\ell(s))) & \text{for $\nu$ in A1--A3,} \\
    \mathcal F_+(\chi_\nu(\ell(s))) & \text{for $\nu$ in B1--B3,} \\
    \mathcal F_{++}(\chi_\nu(\ell(s))) & \text{for $\nu = \textrm{C1}$.} \\
\end{cases}
\]
This is a continuous
real function on $[0,1]$: the substituted expressions are rational,
and their denominators are positive along the line segment. By
\eqref{eq:toshow-new} 
\[
 f(s)^2>0,\qquad 0\leq s\leq1.
\]
Therefore $f$ never vanishes. Since $f(1)>0$ by
Table~\ref{tab:positive-centers}, the intermediate value theorem
implies $f(s)>0$ for all $s$, and in particular $f(0)>0$.
This finishes the proof.

\section{Reduction to the three-point inequality}\label{sec:reduction}

In this section we prove the weak Hellinger inequality \eqref{eq:weak-hellinger} from the three-point inequality, Proposition \ref{prop:three-point}, following the argument in  \cite{ABCNJ17}.  

Set $U=g(Y)$ and $V=f(X)$. The pair $(U,V)$ of random variables is characterized by three parameters, but there is still a fourth parameter $\rho$ and the four parameters are intertwined in some delicate way.
We can get rid of $\rho$ by using a data processing inequality involving the hypercontractivity parameter at zero \cite{ABCNJ17,KA15}
\begin{equation}\label{eqn:dataprocessing}
s_0(g(Y),f(X)) \le s_0(Y, X)=\rho^2,
\end{equation}
where $s_0(U, V)=\lim_{p\to 0+} s_p(U,V)$ and $s_p(U,V)$ for $p\in (0,1)$ denotes the (reverse) hypercontractivity parameter of $(U,V)$ defined as
\[ s_p(U,V)=\inf\big\{\tfrac{1-q}{1-p}:\ p\le q\le 1,\quad
\|\mathbf{E}[b(V)\mid U]\|_p\ge\|b(V)\|_q
\ \text{for every }b\ge0\big\},
\]
Thus it suffices to prove the stronger inequality
\begin{equation}\label{eqn:stronger}
h(\mathbf{E} V) - \mathbf{E}\,h(\mathbf{E}(V|U)) \le 1 - \sqrt{1-s_0(U,V)}.
\end{equation}
This is Conjecture 5 in \cite{ABCNJ17}.
Now set 
\[s=\mathbf{P}(U=1), c = \mathbf{P}(V=1 | U=-1), d=\mathbf{P}(V=-1 | U=1).\]
If $f$ or $g$ is constant, or if $|\rho|=1$, then there is nothing to show, thus we may assume without loss of generality that $s,c,d\in (0,1)$.

For $t\in [0,1]$ define
\[ J_{s,c,d}(t) = \min_{\substack{a,b\in[0,1]\\sa+\bar sb=t}} s\,D(a\mid \bar d)+\bar s D(b\mid c). \]
With this we can write the formula for $s_0(U,V)$ from \cite{BN16,Kam15} as
\[ 1-s_0(U,V)=
\inf_{\substack{t\in[0,1],\\t\not=s\,\bar d+\bar s\,c}} \frac{D(t\mid s\,\bar d+\bar s\,c)}{J_{s,c,d}(t)}. \]
Note that $s\,\bar d+\bar s\,c=\mathbf{P}(V=1)$. Next estimate the infimum from above by the value at $t=s\,\bar c+\bar s\,d$, where $J_{s,c,d}(s\,\bar c+\bar s\,d)=s D(\bar c \mid \bar d) + \bar s D(d \mid c)$. This gives
\[ \sqrt{1-s_0(U,V)} \le R(s,c,d), \]
whence \eqref{eqn:stronger} follows from \eqref{eq:three-point}.

\begin{remark}\label{rem:mutualinf}
By convexity of the monotone function $\varphi(t)=H_2((1-\sqrt{1-t^2})/2)$, $t\in [0,1]$, and the inequality \eqref{eqn:stronger} we also obtain
\begin{equation}\label{eqn:weakmutualinfs0} I(U; V)\le 1 - H_2\Big(\frac{1-\sqrt{s_0(U;V)}}2\Big),
\end{equation}
which again implies \eqref{eqn:weakmutualinf} by the data processing inequality \eqref{eqn:dataprocessing}.
\end{remark}

\end{document}